\documentclass[11pt,letterpaper,onecolumn]{article}
\usepackage[T1]{fontenc}
\usepackage{fullpage, dsfont, amsthm, amssymb, amsmath,stmaryrd,authblk}

\usepackage{graphicx,color,tikz}
\usepackage{hyperref}

\usepackage[capitalise]{cleveref}
\usetikzlibrary{positioning,shapes,fit,arrows}
\usetikzlibrary{decorations.markings}

\usepackage{mymacros}

\title{Substochastic Monte Carlo Algorithms}

\author[1]{Michael Jarret \thanks{mjarret@pitp.ca}}
\affil[1]{\footnotesize\PITP}
\author[2,3,4]{Brad Lackey \thanks{bclackey@umd.edu}}
\affil[2]{\footnotesize Joint Center for Quantum Information and Computer Science, University of Maryland, College Park, USA}
\affil[3]{\footnotesize Departments of Computer Science and Mathematics, University of Maryland, College Park, USA}
\affil[4]{\footnotesize Mathematics Research Group, National Security Agency, Ft. G. G. Meade, MD, USA}

\begin{document}
    \maketitle

\newtheorem{theorem}{Theorem}
\newtheorem{lemma}[theorem]{Lemma}
\newtheorem{proposition}[theorem]{Proposition}
\newtheorem{corollary}[theorem]{Corollary}
\newtheorem{definition}[theorem]{Definition}

\renewcommand{\discussion}[1]{}
\newcommand{\half}{\frac{1}{2}}

\begin{abstract}
In this paper we introduce and formalize Substochastic Monte Carlo (SSMC) algorithms. These algorithms, originally intended to be a better classical foil to quantum annealing than simulated annealing, prove to be worthy optimization algorithms in their own right. In SSMC, a population of walkers is initialized according to a known distribution on an arbitrary search space and varied into the solution of some optimization problem of interest. The first argument of this paper shows how an existing classical algorithm, ``Go-With-The-Winners'' (GWW), is a limiting case of SSMC when restricted to binary search and particular driving dynamics. 

Although limiting to GWW, SSMC is more general. We show that (1) GWW can be efficiently simulated within the SSMC framework, (2) SSMC can be exponentially faster than GWW, (3) by naturally incorporating structural information, SSMC can exponentially outperform the quantum algorithm that first inspired it, and (4) SSMC exhibits desirable search features in general spaces. Our approach combines ideas from genetic algorithms (GWW), theoretical probability (Fleming-Viot processes), and quantum computing. Not only do we demonstrate that SSMC is often more efficient than competing algorithms, but we also hope that our results connecting these disciplines will impact each independently. An implemented version of SSMC has previously enjoyed some success as a competitive optimization algorithm for Max-$k$-SAT \cite{Jarret2016AdiabaticMethods,LackeyGithub}.
\end{abstract}

\maketitle

\newpage
\section{Introduction}
In 1997, Aldous and Vazirani introduced ``Go-with-the-Winners'' (GWW) algorithms capable of exponentially outperforming depth-first search at finding the deepest nodes of layered graphs \cite{AV94}. In the original exposition, these algorithms served both as a mathematically accessible model of particular aspects of genetic algorithms and as a foil for simulated annealing. Recently, in \cite{Jarret2016AdiabaticMethods} we introduced a numerical algorithm which we called Substochastic Monte Carlo (SSMC) as a continuous time optimization algorithm, and claimed that numerical simulations of SSMC share many features with GWW. However, SSMC is more general than GWW; where GWW is restricted to layered graphs obeying certain constraints, SSMC can be used in general search spaces.

Our primary focus in \cite{Jarret2016AdiabaticMethods} was creating a classical foil to quantum annealing (QA). \footnote{SSMC may be interpreted as a form of classically simulating QA \cite{Albash2016,andriyash2017can,vinci2017non,nishimori2016exponential,crosson2016simulated}. In this work, however, we caution against that interpretation and treat it as a competing classical algorithm.} In our prior work, we focused on SSMC as a population algorithm that approximates a classical adiabatic evolution in order to solve optimization problems. In this paper, we rigorously analyze SSMC as an optimization algorithm in its own right. In particular, we show that under appropriate ``driving dynamics,'' SSMC converges to GWW. Nevertheless, we show that under our usual driving dynamics, SSMC can be exponentially faster than GWW. Furthermore, as something of a corollary to this result, we show that SSMC can actually be exponentially faster than the equivalent QA algorithm. That is, (1) SSMC is more general than GWW, (2) GWW can be efficiently simulated within the SSMC framework, (3) SSMC can be exponentially faster than GWW, (4) by incorporating structural information, SSMC can exponentially outperform the quantum algorithm that inspired it \textit{with the same annealing schedule}, and (5) in more general search settings, SSMC exhibits desirable search characteristics, such as gradient descent against biases.

Our approach combines ideas from genetic algorithms (GWW), theoretical probability (Fleming-Viot processes), and quantum computing (QA). Not only do we demonstrate that SSMC is often more efficient than competing algorithms, but we also hope that our results connecting these disciplines will impact each independently. An implemented version of SSMC has previously enjoyed some success as a competitive optimization algorithm for Max-$k$-SAT \cite{Jarret2016AdiabaticMethods,LackeyGithub}.

In \cref{sec:SSMC}, we introduce SSMC. In \cref{sec:GWW}, we introduce GWW and show convergence of SSMC to GWW. In \cref{sec:NumApprox}, we introduce our usual numerical approach to simulating SSMC. In \cref{sec:examples}, we provide examples that exhibit exponential speedups for SSMC over GWW and QA. Finally, in \cref{sec:general}, we demonstrate that SSMC is capable of performing desirable search procedures (in this case, we focus on gradient descent against a biased walk) in more general search spaces. 

\section{Substochastic Monte Carlo}\label{sec:SSMC}

Substochastic Monte Carlo (SSMC) refers to numerical algorithms based on simulating a renormalized continuous time substochastic process. Conceptually, these are similar to Fleming-Viot processes for approximating the dynamics of an absorbing Markov chain \cite{fleming1979some}. In the language of Fleming-Viot processes, SSMC diffuses a population of walkers on a graph while an objective function governs the rate at which a walker is absorbed, or ``dies.'' Each absorbed walker is repopulated by transporting it to the site of a randomly selected surviving walker.

The underlying dynamics of an SSMC instance is governed by a time-dependent transition rate matrix $H(t)$ through the diffusion equation,
\begin{equation}\label{eqn:heat}
    \begin{cases}
        \frac{d \psi}{d t}(t;y) = - \sum_x H(t)_{y,x}\psi(t;x) \\
        \psi(0;x) = \psi_0(x).
    \end{cases}
\end{equation} 
Here, $\psi_0$ is some ideal initial distribution and $\psi(T)$ encodes the solution to some optimization problem. In the setting of quantum annealing \cite{farhi2000quantum,farhi2002quantum}, $H$ would typically take the form
\begin{equation}\label{eqn:quantum}
    H(t) = a(t)L + b(t)W.
\end{equation}
Here $L$ is a (weighted) graph Laplacian of some search graph $G$ with vertices $V(G)$. Namely, for vertices $x,y\in V(G)$ the entry $L_{yx} = -w_{yx}$, where the edge weight $w_{yx}$ is the transition rate from $x$ to $y$; along the diagonal, $L_{xx} = \sum_{y\in V(G)\setminus\{x\}} L_{yx}$. The matrix $W$ is diagonal with entries $W_{xx} = E_x$, where the goal is to find $x\in V(G)$ that minimize $\{E_x\}$. Typically one takes $a(t) = 1-t/T$ and $b(t)=t/T$ and thus $H(t)$ simply interpolates between some ``driving dynamics'' encoded by $L$ and the optimization problem $W$. Such an interpolation is usually called an ``annealing schedule.'' The hope in adiabatic optimization is that if one takes $T$ to be sufficiently large, then the solution to \cref{eqn:heat} at each time $t$ remains close to the lowest eigenvector of $H(t)$. At time $t = T$ this is supported on elements of $V(G)$ that minimize $\{E_x\}$, hence solving the problem with bounded probability.

In general, \cref{eqn:heat} defines a ``substochastic'' process as follows. The initial site of this process is a random variable $X(0)$ whose distribution is governed by the initial distribution $\psi_0$:
\begin{equation*}
    \Pr{X(0) = x} = \psi_0(x).
\end{equation*}
At times $t>0$ the site of the process $X(t)$ has distribution $\psi(t)$, the solution of \cref{eqn:heat}. Owing to the objective matrix $W$, this process is substochastic. For instance, if $H = L + W$ does not depend on time, then the solution is given by the matrix exponential
\begin{equation*}
    \psi(t;y) = \sum_x (e^{-(L+W)t})_{y,x}\psi(0;x)
\end{equation*}
and the transition probabilities for the process are the entries of the matrix
\begin{equation}\label{eqn:sample-transition}
    \Pr{X(t) = y\:|\: X(0) = x} = (e^{-(L+W)t})_{y,x}.
\end{equation} 

Unlike a stochastic process where all probabilities sum to $1$, here $\sum_{y\in V(G)} (e^{-(L+W)t})_{y,x} \leq 1$ and so some probability may be lost over time. In order to recover a stochastic process, we introduce an absorbing site or ``cemetery'', denoted as $\infty$. We can make the substochastic process stochastic by including an appropriate transition rate to $\infty$. Continuing the example from \cref{eqn:sample-transition}, we define
\begin{equation*}
    \Pr{X(t) = \infty\:|\: X(0) = x} = 1 - \sum_{y\in V(G)} (e^{-(L+W)t})_{y,x},
\end{equation*}
and so $\{X(t)\}_{t\geq 0}$ defines a stochastic process on $V(G)\cup\{\infty\}$. A substochastic Monte Carlo algorithm is a numerical simulation of the conditional, or renormalized, distribution on $x\in V(G)$:
\begin{equation}\label{eqn:ssmc}
    \Pr{X(t) = x \:|\: X(t) \not=\infty} = \frac{\Pr{X(t) = x}}{\sum_{x'\in V(G)} \Pr{X(t) = x'}}.
\end{equation}

Note that shifting the objective $W \mapsto W + a I$ uniformly scales the transition rate to any non-cemetery state by $e^{-at}$. But as seen in \cref{eqn:ssmc}, this scaling cancels in numerator and denominator and hence SSMC algorithms are invariant under such shifts, as long the process remains substochastic.

In any case, the likelihood of being in the state $\infty$ typically increases exponentially in $t$ and so directly simulating the process of \cref{eqn:ssmc} is generally not viable. In this work, we take our cue from Fleming-Viot processes generalized to time-dependent dynamics \cite{fleming1979some}. We generate a population of $N$ particles, or ``walkers,'' according to an initial distribution $\psi_0$. Each walker then moves independently according to the process law $\{X(t)\}_{t\geq 0}$. However, whenever the process would have a walker die, that is transition to $\infty$, it instead moves immediately to the site of another randomly selected walker. Note that in \cite{Jarret2016AdiabaticMethods}, we described a similar, but distinct SSMC algorithm.

As a process, SSMC is described by a vector $\xi(t) = (\xi_1(t), \dots, \xi_N(t))$ where $\xi_i(t)$ indicates the site of walker $i$ at time $t$. Let us write $\theta(t)_x = \sum_{y} H(t)_{y,x}$ for the ``death'' rate at time $t$ and site $x$. The number of walkers at $x$ is given by the statistic
$$\eta(t;x) = \sum_{i=1}^N \mathds{1}_{\{\xi_i = x\}},$$
whose dynamics is given by the nonlinear equation
\begin{equation*}
    \frac{d\eta}{dt}(t;x) = \sum_{y\not= x} (H(t)_{y,x}\eta(t;x) - H(t)_{x,y}\eta(t;y) + \tfrac{1}{N-1}\eta(t;x)\eta(t;y) \theta(t)_y) - \eta(t;x)\theta(t)_x.
\end{equation*}
We will also need the empirical distribution $m_t(x) = \frac{1}{N}\eta(t;x)$. The Fleming-Viot literature is primarily focused on the case when the generator $H$ is time-independent and defines an irreducible, absorbing Markov chain. In this case, for finite spaces, $m_t(x) \to \psi(t;x)$ as $N$ gets large and $\psi(t;x)$ converges to the lowest eigenvector of $H$ exponentially quickly in $t$ \cite{Collet2013Quasi-StationaryDistributions,asselah2011quasistationary,cloez2016quantitative}.

One goal of this paper is to demonstrate that we can recover discrete-time algorithms in a limit, by taking SSMC using a discontinuous schedule as follows:
\begin{equation}\label{eqn:sched2}
    H(t) = \sum_{j=1}^{T} \mathbbm{1}_{[j-1,j)}(t)L_j + W.
\end{equation}
Here, $L_j$ is the Laplacian of a subgraph $G_j$ of the search graph, where the algorithm will search during time $t\in [j-1,j)$. For example, in a tree search one may choose $G_j$ to be the subtree consisting only of depth $j-1$ and $j$ vertices and the edges connecting them. We assume that the initial distribution is supported in $G_1$, and that $V(G_j) \cap V(G_{j+1}) \not= \emptyset$. Note that at times $t \leq j$, the only nonzero transitions weights are to nodes in one of the subgraphs $G_1, \dots, G_j$. Hence, we have proven the following lemma.

\begin{lemma}\label{lemma:no-future}
    At stage $j$, for any node $x \not\in \bigcup_{i=1}^j V(G_i)$ we have
    $$\Pr{X(j) = x \:|\: X(j) \not= \infty} = 0.$$
\end{lemma}

In the next section we analyze the process $\Pr{X(j) = x \:|\: X(j) \not= \infty}$ for a search tree and compare this to the ``Go-with-the-Winners'' algorithm. In \Cref{sec:NumApprox}, we illustrate how SSMC can be simulated with a population of walkers $\xi(t)$ as above, and in \Cref{sec:examples} provide examples where SSMC has exponential speedup over Go-with-the-Winners.

\section{Go-with-the-winners}\label{sec:GWW}

The ``Go-with-the-Winners'' algorithms (GWW), introduced to study some dynamics of genetic algorithms, were formulated in terms of search algorithms for finding a maximal depth leaf of a tree \cite{AV94}. In that work, ``Algorithm 1'' is closest in spirit to a Fleming-Viot process, which we also refer to as \Cref{alg:GWW1*}, however we specifically use the description from \cite{vazirani1999go}.

\begin{alg}\label{alg:GWW1*}
    Repeat the following strategy:

    Let $\xi(j)$ be the set of walkers at stage $j$. If all $\xi_i(j)$ are leaves, output some random $\xi_i(j)$. Otherwise, let $U$ be the indices of walkers at nonleaves. For each $i\in U$ let $\xi_i(j+1)$ be a randomly selected child of $\xi_i(j)$. Then for each $i \not\in U$, choose a random $k\in U$ and set $\xi_i(j+1) = \xi_k(j+1)$. 
\end{alg}

We can also define a slight variant of \Cref{alg:GWW1*}, where any walker at a leaf will jump \emph{uniformly to nodes occupied by other walkers}. We will see in a moment that these dynamics are also replicated by SSMC.

\begin{ralg}{alg:GWW1*}\label{alg:GWW1}
    Repeat the following strategy:

    Let $\xi(j)$ be the set of walkers at stage $j$. If all $\xi_i(j)$ are leaves, output some random $\xi_i(j)$. Otherwise, let $U$ be the indices of walkers at nonleaves. For each $i\in U$ let $\xi_i(j+1)$ be a randomly selected child of $\xi_i(j)$. Let $V$ be the nodes at level $j+1$ occupied by at least one walker; for $i \not\in U$, choose a random $x\in V$ and set $\xi_i(j+1) = x$. 
\end{ralg}

Although this distinction seems minor, these two algorithms will occasionally exhibit drastically different behavior. Nonetheless we will provide examples where SSMC achieves speedups over both in the following sections. As an aside, \cite[Algorithm 2]{AV94} is a method for approximating \Cref{alg:GWW1*} where the jump process is replaced by a birth/death process, each walker at a leaf dies, and the expected total population is maintained by spawning walkers at nonleaves. In \cite{Jarret2016AdiabaticMethods}, we provided an analogue to this method for SSMC.

To recover behavior like \Cref{alg:GWW1*} with SSMC, we first need to decompose the tree into subgraphs $\{G_j\}$. Let us write $V_j$ for the nodes of $G$ of depth $j$. Stage $j$ of the algorithm involves walkers moving from nodes at depth $j-1$ to those at depth $j$, so we set the vertices of our subgraph as $V(G_j) = V_{j-1}\cup V_j$ and the edges of $G_j$ are all the connections between these nodes in the tree. To ensure walkers descend, we set an objective function whose value at nodes in $V_{j-1}$ is much larger than at those in $V_j$. For ease of analysis we take this to be $E$, independent of the stage. That is, if the depth of the tree is $D$, then the root has objective value $DE$, the nodes at depth $1$ have value $(D-1)E$, and so on. However, from \cref{lemma:no-future} the process is supported on vertices no deeper than $j$. Above we noted that shifting the objective value does not affect SSMC, and so we may assume that during stage $j$ the objective value vanishes on nodes of $V_j$, takes value $E$ on notes of $V_{j-1}$ and so on.

Let us write  $\{X_E(t)\}_{t\geq 0}$ for the SSMC process \cref{eqn:ssmc}, where the dynamics are given by \cref{eqn:heat,eqn:sched2} using the graph and objective function as described above. 

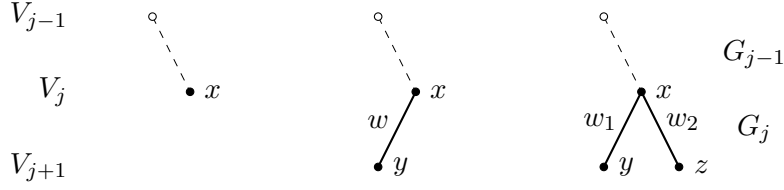
\begin{figure}[t]
\begin{center}
\begin{tikzpicture}
\draw[anchor=east] (-1,2) node {$V_{j-1}$};
\draw[anchor=east] (-1,1) node {$V_j$};
\draw[anchor=east] (-1,0) node {$V_{j+1}$};
\draw[very thin, dashed] (0,2) -- (0.5,1);
\draw[fill=white] (0,2) circle (0.05);
\draw[fill=black] (0.5,1) circle (0.05);
\draw (0.8,1) node {$x$};
\draw[very thin, dashed] (3,2) -- (3.5,1);
\draw[fill=white] (3,2) circle (0.05);
\draw[thick] (3.5,1) -- (3,0);
\draw (3,0.6) node {$w$};
\draw[fill=black] (3.5,1) circle (0.05);
\draw[fill=black] (3,0) circle (0.05);
\draw (3.8,1) node {$x$};
\draw (3.3,0) node {$y$};
\draw[very thin, dashed] (6,2) -- (6.5,1);
\draw[fill=white] (6,2) circle (0.05);
\draw[thick] (6.5,1) -- (6,0);
\draw[thick] (6.5,1) -- (7,0);
\draw[fill=black] (6.5,1) circle (0.05);
\draw[fill=black] (6,0) circle (0.05);
\draw[fill=black] (7,0) circle (0.05);
\draw (5.95,0.6) node {$w_1$};
\draw (7.05,0.6) node {$w_2$};
\draw (6.8,1) node {$x$};
\draw (6.3,0) node {$y$};
\draw (7.3,0) node {$z$};
\draw (8,1.5) node {$G_{j-1}$};
\draw (8,0.5) node {$G_j$};
\end{tikzpicture}
\end{center}
\caption{Three types of components in the subgraph $G_j$ in a binary tree.}\label{fig:graphs}
\end{figure}

\begin{lemma}\label{lemma:compute}
    Suppose at time $j$ the process is at a node $x\in G$. Then:
    \begin{enumerate}
        \item if $x \in V_k$ for $k< j$ or $x \in V_j$ is a leaf of $G_j$ then the probability the process has not died by time $j+1$ is negligible in $E$;
        \item if $x\in V_j$ is not a leaf of $G_j$, then the probability the process is at $x$ at time $j+1$ is $\bigO(\frac{1}{E^2})$;
        \item if $y\in V_{j+1}$ is any child of $x$ with transition rate $w$ on the edge $(x,y)$, then the probability the process is at $y$ at time $j+1$ is $\frac{w}{E}e^{-w} + \bigO(\frac{1}{E^2})$ (independent of the number of children $x$ has).
    \end{enumerate}  
\end{lemma}
\begin{proof}
%    See \cref{sec:proof-of-computational-lemma}.
  
    The subgraph $G_j$ decomposes into connected components of three types as to whether a node at depth $j$ is (i) a leaf, (ii) one child, or (iii) has multiple children at depth $j$. Note that nodes at layers $k < j$ are disconnected from subgraph $G_j$. See \cref{fig:graphs}.
  
    If either $x\in V_k$ for $k<j$ or $x\in V_j$ is a leaf, then the Laplacian component is just $L = 0$. The (relative) objective value is $(j-k+1)E$, and therefore the transition rate matrix is also $H = (j-k+1)E$. After time $t\in [0,1)$, we have
    \begin{equation}\label{eqn:ssmc-transition00}
        \mathrm{Pr}\{X_E(j+t) = x \:|\: X_E(j) = x\} =  e^{-(j-k+1)Et}.
    \end{equation}
    In particular the probability the process survives to time $j+1$ is $e^{-(j-k+1)E}$.

    If $x\in V_j$ has a single child $y$ at depth $j+1$, and $w$ is the weight attached to the edge $(x,y)$, then our transition rate matrix on this component is
    \begin{equation*}
        H = \left(\begin{array}{cc} w & -w \\ -w & w\end{array}\right) + \left(\begin{array}{cc} E & 0 \\ 0 & 0\end{array}\right) = \begin{pmatrix} E + w & -w \\ -w & w\end{pmatrix}.
    \end{equation*}
    We can exponentiate this matrix in closed form to compute the transition probabilities at any time $t$, see \cref{sec:computations}, however since we are only interested in the limit of large $E$, we will work asymptotically in $\epsilon = \frac{w}{E}$. Diagonalizing the matrix $\tilde{H} = \frac{1}{E}H$ to order $\epsilon^3$ we have:
    \begin{equation*}
        \tilde{H}
        \begin{pmatrix} 1 & \epsilon \\ -\epsilon & 1 \end{pmatrix} =
        \begin{pmatrix} 1 & \epsilon \\ -\epsilon & 1 \end{pmatrix}
        \begin{pmatrix} 1 + \epsilon + \epsilon^2 & 0 \\ 0 & \epsilon - \epsilon^2 \end{pmatrix} + \bigO(\epsilon^3).
    \end{equation*}

    To derive the action of $e^{-Ht} = e^{-E\tilde{H}t}$ we compute
    $$
    e^{-Ht}\begin{pmatrix} 1 & \epsilon \\ -\epsilon & 1 \end{pmatrix} =
    \begin{pmatrix} 1 & \epsilon \\ -\epsilon & 1 \end{pmatrix}
    \begin{pmatrix} e^{-(E + w)t + \bigO(\epsilon)} & 0 \\ 0 & e^{-wt + \bigO(\epsilon)} \end{pmatrix} + \bigO(\epsilon^2).
    $$
    Notice that
    $$
    \begin{pmatrix} 1 & \epsilon \\ -\epsilon & 1 \end{pmatrix}
    \begin{pmatrix} 1 \\ \epsilon \end{pmatrix} = \begin{pmatrix} 1 \\ 0 \end{pmatrix} + \bigO(\epsilon^2).
    $$
    And so
    \begin{align*}
        e^{-Ht}\begin{pmatrix} 1 \\ 0 \end{pmatrix} &= \begin{pmatrix} 1 & \epsilon \\ -\epsilon & 1 \end{pmatrix} \begin{pmatrix} e^{-(E + w)t + \bigO(\epsilon)} \\ \epsilon e^{-wt + \bigO(\epsilon)} \end{pmatrix} + \bigO(\epsilon^2)\\
        &=\begin{pmatrix} e^{-(E + w)t + \bigO(\epsilon)} \\ \epsilon (e^{-wt + \bigO(\epsilon)} - e^{-(E + w)t + \bigO(\epsilon)})\end{pmatrix} + \bigO(\epsilon^2)
    \end{align*}
    That is,
    \begin{align}
       \Pr{X_E(j+t) = x \:|\: X_E(j) = x} &= e^{-(E+w)t + \bigO(1/E)} + \bigO(\tfrac{1}{E^2}),\label{eqn:transition10}\\
       \Pr{X_E(j+t) = y \:|\: X_E(j) = x} &= \frac{w}{E}e^{-wt} + \bigO(\tfrac{1}{E^2}).\label{eqn:transition11}
    \end{align}
    Taking $t=1$ gives the desired results in this case.

    For a node $x\in V_j$ with $d$ children, $y_1, \dots, y_d \in V_{j+1}$, and associated edge weights $w_1, \dots, w_d$, the transition rate matrix is given by
    \begin{equation*}
       H = \begin{pmatrix} 
        E + w_1 + \cdots + w_d & -w_1 & \cdots & -w_d \\
       -w_1 & w_1 & \cdots & 0 \\
       \vdots & & \ddots \\
       -w_d & 0 & \cdots & w_d\end{pmatrix}
    \end{equation*}
    We introduce the notation $b_1 = \sum_{j=1}^d w_j$ and $b_2 = \sum_{j=1}^d w_j^2$.
    
    Let
    \begin{align*}
        \lambda_0 &= 1 + \frac{b_1}{E} + \frac{b_2}{E^2} + \bigO(\tfrac{1}{E^3})\\
       v_0 &= \begin{pmatrix} 1 \\ -\frac{w_1}{E} + \frac{w_1(b_1 - w_1)}{E^2}) \\ \vdots \\ -\frac{w_d}{E} +  \frac{w_d(b_1-w_d)}{E^2})\end{pmatrix} + \bigO(\tfrac{1}{E^3}).
    \end{align*}
    Then a direct computation shows both $\frac{1}{E}Hv_0$ and $\lambda_0v_0$ equal
    \begin{equation*}
        \begin{pmatrix} 1 \\ 0 \\ \vdots \\ 0 \end{pmatrix} + \frac{1}{E}\begin{pmatrix} b_1 \\ -w_1\\ \vdots \\ -w_d\end{pmatrix} + \frac{1}{E^2}\begin{pmatrix} b_2 \\ -w_1^2 \\ \vdots \\ -w_d^2\end{pmatrix} + \bigO(\tfrac{1}{E^3})
    \end{equation*}
    and so $Hv_0 = E\lambda_0v_0 + \bigO(\frac{1}{E^2})$.

    Let us assume for the moment that all the edge weights are distinct. Then for each $j=1, \dots, d$ define
    \begin{align*}
        \lambda_j &= \frac{w_j}{E} - \frac{w_j^2}{E^2} + \bigO(\tfrac{1}{E^3})\\
        v_j &= \begin{pmatrix} \frac{w_j}{E} + \frac{w_j^2}{E^2}\sum_{k\not=j} \frac{w_k}{w_j-w_k} \\ -\frac{1}{E}\frac{w_jw_1}{w_j-w_1} \\ \vdots \\ 1 \\ \vdots \\ -\frac{1}{E}\frac{w_jw_d}{w_j-w_d}\end{pmatrix} + \bigO(\tfrac{1}{E^3})
    \end{align*}
    Then a similar computation shows both $\frac{1}{E}Hv_j$ and $\lambda_jv_j$ equal
    \begin{equation*}
        \frac{1}{E}\begin{pmatrix} 0 \\ 0\\ \vdots \\ w_j \\ \vdots \\ 0\end{pmatrix} + \frac{1}{E^2}\begin{pmatrix} w_j^2 \\ -w_j^2w_1/(w_j-w_1) \\ \vdots \\ -w_j^2 \\ \vdots \\ -w_j^2w_d/(w_j-w_d)\end{pmatrix} + \bigO(\tfrac{1}{E^3})
    \end{equation*}
    and so $Hv_j = E\lambda_jv_j + \bigO(\frac{1}{E^2})$.

    Exponentiation, as in the proof of \cref{lemma:compute}, gives
    \begin{align*}
        e^{-Ht}v_0 &= e^{-(E + b_1)t + \bigO(1/E)}v_0 +\bigO(\tfrac{1}{E^2}),\\
        e^{-Ht}v_j &= e^{-w_jt + \bigO(1/E)}v_j + \bigO(\tfrac{1}{E^2})\text{ for $j=1, \dots, d$.}
    \end{align*}
    We also see
    \begin{equation*}
        \begin{pmatrix} 1 \\ 0 \\ \vdots \\ 0\end{pmatrix} = v_0 + \frac{w_1}{E}v_1 + \cdots \frac{w_d}{E}v_d + \bigO(\tfrac{1}{E^2}),
    \end{equation*}
    and therefore
    \begin{align*}
        e^{-Ht} \begin{pmatrix} 1 \\ 0 \\ \vdots \\ 0\end{pmatrix} &= e^{-(E + b_1)t +\bigO(1/E)}v_0 + \sum_{j=1}^d \frac{w_1}{E}e^{-w_jt +\bigO(1/E)} v_j + \bigO(\tfrac{1}{E^2})\\
        &= \frac{1}{E}\begin{pmatrix} 0 \\  w_1e^{-w_1t +\bigO(1/E)} \\ \vdots \\ w_de^{-w_dt +\bigO(1/E)} \end{pmatrix} + \bigO(\tfrac{1}{E^2}).
    \end{align*}
    Expressing this in terms of transition probabilities has
    \begin{align*}
        \Pr{X_E(j+t) = x \:|\: X_E(j) = x} &= \bigO\left(\tfrac{1}{E^2}\right),\\
        \Pr{X_E(j+t) = y_j \:|\: X_E(j) = x} &= \frac{w_j}{E}e^{-w_jt} + \bigO\left(\frac{1}{E}\right),
    \end{align*}
    for each $j=1, \dots, d$.

    Note that the resulting transition probabilities vary smoothly with small perturbations in the weights $w_j$. Therefore if not all the weights are distinct, one perturbs them slightly into distinct values, and the above analysis produces the desired expression. Again taking $t=1$ proves the lemma.
\end{proof}

\begin{theorem}\label{thm:limit-process}
    Let $G$ be a search tree with transition rates $w_{xy} > 0$ from parent node $x$ to child node $y$ and objective function given by $E$ times the height of a node. Define a discrete-time process $\{X(j)\}_{j=0,1,2,\dots}$ to be given by $\Pr{X(0) = \text{root}} = 1$ and for $j \geq 0$: 
    \begin{equation}\label{eqn:transition}
        \Pr{X(j+1) = y} = \frac{\sum_{x\in V_j} a_{xy}\Pr{X(j) = x}}{\sum_{y'\in V_{j+1}}\sum_{x\in V_j} a_{xy'}\Pr{X(j) = x}}
    \end{equation}
    where $a_{xy} = w_{xy}e^{-w_{xy}}$ if $y$ is a child of $x$, and $a_{xy} = 0$ otherwise. Then the statistical difference of this process and $\{X_E(t)\}_{t\geq 0}$ is $\delta(X_E(j),X(j)) = \bigO(\frac{1}{E})$.
\end{theorem}
\begin{proof}
%    See \cref{sec:proof-of-limit-theorem}.

    Initially, we have $\Pr{X_E(0) = \text{root}} = \Pr{X(0) = \text{root}} = 1$ and so $\delta(X_E(0),X(0)) = 0$. Inductively, assume that at time $t = j$ we $\delta(X_E(j),X(j)) =\bigO(\frac{1}{E})$. 
    
    We first compute the normalization term $\Pr{X_E(j+1) \not= \infty \:|\: X_E(j)\not=\infty}$. In general this is
    \begin{align*}
        &\Pr{X_E(j+1) \not=\infty\:|\: X_E(j)\not=\infty}\\
        &\quad = \sum_{x\in G} \Pr{X_E(j+1) \not=\infty\:|\:X_E(j) = x}\Pr{X_E(j) = x\:|\: X_E(j)\not=\infty}\\
        &\quad = \sum_{x,y\in G} \Pr{X_E(j+1) =y\:|\:X_E(j) = x}\Pr{X_E(j) = x\:|\: X_E(j)\not=\infty}.
    \end{align*}
    Note that by \cref{lemma:compute}, $\Pr{X_E(j+1) =y\:|\:X_E(j) = x}$ is negligible in $E$ unless $x,y \in G_j$. Even then, $\Pr{X_E(j+1) =y\:|\:X_E(j) = x} =\bigO(\frac{1}{E^2})$ unless $y$ is a child of $x$. So, given $x\in V_j$ with child $y\in V_{j+1}$ \Cref{lemma:compute} gives
    \begin{equation*}
        \Pr{X_E(j+1) =y\:|\:X_E(j) = x} = \frac{w_{xy}}{E}e^{-w_{xy}} +\bigO(\tfrac{1}{E^2}).
    \end{equation*}
    Therefore, since $\delta(X_E(j),X(j)) =\bigO(\frac{1}{E})$ by hypothesis
    \begin{align*}
        &\Pr{X_E(j+1) \not=\infty\:|\: X_E(j+1)\not=\infty}\\
        &\quad = \sum_{y\in V_{j+1}}\sum_{x\in V_j} \frac{a_{xy}}{E} \Pr{X_E(j) = x\:|\: X_E(j) \not=\infty} +\bigO(\tfrac{1}{E^2})\\
        &\quad = \sum_{y\in V_{j+1}}\sum_{x\in V_j} \frac{a_{xy}}{E} \Pr{X(j) = x} +\bigO(\tfrac{1}{E^2}).
    \end{align*}
    
    Now for a general $y\in V_{j+1}$ with parent $x\in V_j$, again using \Cref{lemma:compute} and the inductive hypothesis,
    \begin{align*}
        &\Pr{X_E(j+1) = y\:|\: X_E(j+1)\not=\infty}\\
        &\quad = \frac{\Pr{X_E(j+1) = y\:|\:X_E(j) = x}\Pr{X_E(j) = x\:|\: X_E(j)\not=\infty}}{\Pr{X_E(j+1) \not= \infty \:|\: X_E(j)\not=\infty}}\\
        &\quad = \frac{\left(\frac{a_{xy}}{E} +\bigO(\frac{1}{E^2})\right)\left(\Pr{X(j) = x} +\bigO(\frac{1}{E})\right)}{\sum_{y'\in V_{j+1}}\sum_{x'\in V_j} \frac{a_{x'y'}}{E} \Pr{X(j) = x'} +\bigO(\tfrac{1}{E^2})}\\
        &\quad = \frac{\sum_{x\in V_j} a_{xy}\Pr{X(j) = x}}{\sum_{y'\in V_{j+1}}\sum_{x'\in V_j} a_{x'y'}\Pr{X(j) = x'}} +\bigO(\tfrac{1}{E})\\
        &\quad = \Pr{X(j+1) = y} +\bigO(\tfrac{1}{E}).
    \end{align*}
    Therefore $\delta(X_E(j+1),X(j+1)) =\bigO(\frac{1}{E})$.
\end{proof}

The numerator in \cref{eqn:transition} is just $a_{xy}\mathrm{Pr}\{X(j) = x\}$ where $x$ is the parent of $y$, as all other terms in the sum vanish. For SSMC to match \cref{alg:GWW1*}, where walkers at a nonleaf node transition to child nodes uniformly at random, we take $w_{xy}$ so that $a_{xy} = w_{xy}e^{-w_{xy}} = \frac{1}{2d_x}$ when $y$ is a child of $x$, which has with $d_x$ children. Note the additional factor of $\frac{1}{2}$ in $a_{xy}$ is irrelevant since it cancels in the numerator and denominator of \cref{eqn:transition}, but is necessary since $we^{-w} \leq \frac{1}{2}$. 

On the other hand, if we take $w_{xy}$ constant, we obtain an algorithm closer to \cref{alg:GWW1} where walkers transition to the next layer uniformly in the nodes at that layer. We state this result formally as follows.

\begin{corollary}
    Let $G$ be a search tree whose edge weights are all equal. Then at time $t = j$ the statistical difference between the distribution of $X_E(t)$ and the uniform distribution on the nodes at depth $j$ is $\bigO(\tfrac{1}{E})$.
\end{corollary}
\begin{proof}
    As $\delta(X_E(j),X(j)) = \bigO(\frac{1}{E})$, the result follows if the limiting distribution $X(j)$ is uniform. Again we work inductively. At time $t = 0$ we have $\Pr{X(0) = \text{root}} = 1$, and so the distribution is the uniform at depth $0$. Assume that at time $t = j-1$ for every $x\in V_{j-1}$,
    \begin{equation*}
        \Pr{X(j-1) = x} = \frac{1}{|V_{j-1}|}.
    \end{equation*}
    Let $y$ be a node at depth $j$ and let $x$ be the parent of $y$. Then from \cref{eqn:transition}
    \begin{align*}
        \Pr{X(j) = y} &= \frac{a_{xy}\Pr{X(j-1) = x}}{\sum_{y'\in V_j}\sum_{x'\in V_{j-1}} a_{x'y'}\Pr{X(j-1) = x'}}\\
        &= \frac{a_{xy}/|V_{j-1}|}{\sum_{y'\in V_j}\sum_{x'\in V_{j-1}} a_{x'y'}/|V_{j-1}|} = \frac{a_{xy}}{\sum_{y'\in V_j} a_{x'y'}} = \frac{1}{|V_{j}|}.
    \end{align*}
\end{proof}

\section{The numerical simulation}\label{sec:NumApprox}

The purpose of this section is to present our standard simulation and in \cref{sec:examples} show that these simulations achieve speedups over \Cref{alg:GWW1*} and \Cref{alg:GWW1}. The method of the previous section already demonstrates that if we are able to explicitly choose our transition rates according to the continuous time process, we can expect better performance than what follows. Nonetheless, in arbitrarily large spaces with unknown objective functions, we cannot explicitly compute transition probabilities via matrix exponentials. In this section, we detail a local search method (like that in \cite{Jarret2016AdiabaticMethods} but simpler to analyze and closer to a pure Fleming-Viot process) setting the stage for more powerful future analyses. We model our distribution $\psi$ empirically as a population of walkers $\xi$. In this section, for ease of presentation, we primarily focus on these algorithms when restricted to the schedule of \cref{eqn:sched2}.

Recall that for our process, we have that \cref{eqn:heat} solved by \cref{eqn:sample-transition} for $t\in[j,j+1)$. We consider the first order approximation of this matrix for a small time step $\Delta t$,
\[
    \psi(t+\Delta t) = e^{-(L_j+W)\Delta t}\psi(t) \approx \left(I - \Delta t (L_j + W) + \bigO(\Delta t^2)\right)\psi(t).
\]
Note that when $\Delta t$ is sufficiently small, 
\[
    (I - \Delta t (L_j + W))_{xy} \geq 0 
\]
and
\[
    \sum_x(I - \Delta t (L_j + W))_{xy} \leq 1. 
\]
In other words, there exists a choice of $\Delta t$ such that $T(\Delta t) = (I - \Delta t (L_j + W))$ is always a substochastic transition matrix. Here, $T:V \rightarrow V\cup\{\infty\}$. Also, for $\overline{W}_t = W - I \min_{j=1,\dots,N} W(\xi_j(t))$, we allow $W\mapsto \overline{W}_t$ and consider the corresponding transition matrix $\overline{T}(\Delta t)$. (Since $W$ is a diagonal operator, we use the shorthand $W(x) = W_{x,x}$.)

Now we propose the following simulation:
\begin{alg}\label{alg:SSMC}
    While $t \in [0,T]$ repeat the following procedure: 
    
    Suppose we have a collection of $N$ walkers with configuration $\xi(t) \in V^N$ where by $\xi_i(t)$ we denote the position of walker $i$. Perform the walk prescribed by $\xi(t+\frac{\Delta t}{2}) = \overline{T}(\Delta t) \xi(t)$. If $\xi(t+\frac{\Delta t}{2}) = \xi(t)$, let $\xi(t+1) = \xi(t)$ and increment $t \mapsto \floor{t+1}$. Otherwise, let $U = \left\{i \vert \xi_i(t+\frac{\Delta t}{2}) \neq \infty \right\}$. Perform the transition,
    \[
        \xi_i(t+\Delta t) = \begin{cases}
            \xi_i(t+\frac{\Delta t}{2}) & \text{if $i\in U$} \\
            \xi_k(t+\frac{\Delta t}{2}) & \text{if $i\not\in U$, where $k \in U$ chosen at random.}
        \end{cases}
    \]
    Increment $t \mapsto t+\Delta t$.
\end{alg}

Note that, potentially at the sacrifice of efficiency, we are free to choose $\Delta t$ as small as we like. There always exists a choice of $\Delta t$ such that the probability $\prod_{x \in X} T_{x \mapsto\infty}$ is arbitrarily small, so that we can limit the probability that more than $1$ walker dies. In some sense, this provides us the freedom to simulate the stationary Fleming-Viot process for $N$ walkers, however this will rarely be efficient for optimization.

For our simulations, when restricting to search trees all with edge weight 1, the mapping $W\mapsto \overline{W}_t$ guarantees that at time $t=j$, $\overline{T}_{x\mapsto \infty} = 0$ for $x \in V_{j}$. That is, at the start of stage $j$, all walkers are at most at depth $j$ and hence jump to the cemetery. To be small enough to guarantee substochasticity, $\Delta t$ must scale like $\frac{1}{E}$, where $E$ is the gradient of $W$, defined as in \cref{sec:GWW}. Because we can always rescale $E$, we consider the case that $E$ is sufficiently large such that $\overline{T}_{x \mapsto\infty} \sim 1$ for $x \in V_{j-1}$. In other words, we look at cases where approximately, when the deepest occupied nodes are of depth $j$, all walkers at depths $j' < j$ transition to the cemetery almost surely.

Assuming that the probability that a walker remains in $V_{j-1}$ at stage $j$ is given by $\frac{p_l}{E}$ for some constant $p_l$, the probability that no walkers lag in the depth $D$ algorithm will scale like $(1-\frac{p_l}{E})^{N D} \geq 1-\frac{p_l N D}{E}$. Hence, choosing $E \sim \mathrm{poly}(N D)$, we can condition on the non-occurrence of lagging walkers during the course of the algorithm while only changing the probability of success slightly. Conditioning on this non-occurrence simplifies our analysis to a two stage process, not unlike \cref{alg:GWW1*} and \cref{alg:GWW1}, but where we now have a lazy walk jump process. 

%\begin{ralg}{alg:SSMC}\label{alg:SSMC*}
%    While $t \in [0,T]$ repeat the following procedure: 
%    
%    Suppose we have a collection of $N$ walkers with configuration $\xi(t) \in V^N$ where by $\xi_i(t)$ we denote the position of walker $i$. If $t+\Delta t \geq \floor{t}+1$, return $\xi(t)$. Otherwise, perform the walk prescribed by $\xi(t+\frac{\Delta t}{2}) = \overline{T}(\Delta t) \xi(t)$. If $\xi(t+\frac{\Delta t}{2}) = \xi(t)$, let $\xi(t+1) = \xi(t)$ and increment $t \mapsto \floor{t+1}$. Otherwise, let $X = \left\{i \vert \xi_i(t+\frac{\Delta t}{2}) \neq \infty \right\}$. Perform the transition,
%    \[
%        \xi_i(t+\Delta t) = \begin{cases}
%            \xi_i(t+\frac{\Delta t}{2}) & \text{if $i \in X$} \\
%            \xi_k(t+\frac{\Delta t}{2}) & \text{for $k \in X$ chosen at random.}
%        \end{cases}
%    \]
%    Increment $t \mapsto t+\Delta t$. 
%\end{ralg}
\begin{alg}\label{alg:SSMC*}
    Repeat the following strategy:

    Let $\xi(j)$ be the set of walkers at stage $j$ and $V_j$ the nodes of depth $j$.  If all $\xi_i(j)$ are leaves, output some random $\xi_i(j)$. Let $d(\xi_i)$ be the degree of the node at position $\xi_i$, and let $d_{max}=\max_{i} d(\xi_i)$. Then,
    \[
        \xi_i(j+\tfrac{1}{2}) = 
        \begin{cases}
            \xi_i(j) & \text{with probability $1 - \frac{d(\xi_i)}{d_{max}}$} \\
            \text{a child of $\xi_i(j)$} & \text{each with probability $\frac{1}{d_{max}}$.}
        \end{cases}
    \]
    Then, for each $\xi_i(j+\frac{1}{2})$, let
    \[
        \xi_i(j+1) = \begin{cases}
            \xi_i(j+\tfrac{1}{2}) & \text{if $\xi_i(j+\tfrac{1}{2}) \in V_{j+1}$}\\
            \xi_k(j+\tfrac{1}{2}) & \text{if $\xi_i(j+\tfrac{1}{2}) \notin V_{j+1}$, for a random choice of $k\in \left\{k' \vert \xi_{k'}(j+\tfrac{1}{2}) \in V_{j+1}\right\}$.}
            \end{cases}
    \]
\end{alg}

In the next section, we use \cref{alg:SSMC*} in order to prove a separation between \cref{alg:SSMC} and both \cref{alg:GWW1*} and \cref{alg:GWW1}. That the behavior of SSMC cannot be entirely understood as a lazy walk on a tree will be demonstrated in \cref{sec:general}, by considering general search spaces. 

\section{Speedups}\label{sec:examples}
\subsection{The comb tree}
    
    \begin{figure}\centering
\begin{subfigure}[b]{.4\textwidth}\centering
    \begin{tikzpicture}[scale=0.7,every node/.style={draw, circle, inner sep=0pt,minimum size=4pt},
    level 1/.style={sibling distance=7em},
    level 2/.style={sibling distance=7em}
    ]
        \node {}
        child {node {}
            child {node {}
                child { node {}
                    child { node {} edge from parent[dotted]}
                    %child { node [draw=none] {} edge from parent[draw=none]}
                }
                %child { node [draw=none] {} edge from parent[draw=none]}
            }
            %child { node [draw=none] {} edge from parent[draw=none]}
            %child {node [style={size=0pt}] {} edge from parent [] }
        }
        child {node {}
            %child {node {}
            %}
            %child {node {}
                child{ node {}}
                child{ node {} 
                    child{ node {}}
                    child{ node {} edge from parent[dotted] {
                    child{ node [solid] {} edge from parent[solid]}
                    child{ node [solid] {} edge from parent[solid]}}
                    }
                }
            %}
        };
    \end{tikzpicture}\caption{\label{fig:comb}}
\end{subfigure}
\begin{subfigure}[b]{.4\textwidth}\centering
    \begin{tikzpicture}[scale=0.7,every node/.style={draw, circle, inner sep=0pt,minimum size=4pt},
    level/.style={sibling distance=7em}
    ]
        \node {}
        child {node {}
                child { node {}
                    child { node {}
                        child { node {} edge from parent[dotted] 
                            %child { node[solid] {} edge from parent[solid]
                                %child { node[solid] {} edge from parent[dotted]}
                            %}
                        }
                    }
                }
            }
        child {node {}
                child {node {}
                    child { node {}
                        child { node {} edge from parent[dotted] 
                            %child { node[solid] {} edge from parent[solid]
                                %child { node[solid] {} edge from parent[dotted]}
                            %}
                        }
                    }
                }
                child{ node {} 
                    child{ node {} 
                            child { node {} edge from
                                parent[dotted]
                                    %child { node[solid] {} edge from parent[solid]
                                        %child { node[solid] {} edge from parent[dotted]}
                                    %}
                            }
                         }
                    child{ node {} edge from parent[dotted] {
                        child{ node [solid] {} edge from     parent[solid] 
                                %child { node[solid] {} edge from parent[solid] %child { node[solid] {} edge from parent[dotted]}
                            %}
                        }
                        child{ node [solid] {} edge from parent[solid] %child { node[solid] {} edge from         parent[solid] child { node[solid] {} edge from parent[solid]
                            %child { node[solid] {} edge from parent[solid]}            
                        }}%}}
                        }
                    }
            };
    \end{tikzpicture}\caption{\label{fig:waterfall}}
\end{subfigure}
\caption{\cref{fig:comb} shows a simple comb graph with a designated root and all teeth of length $1$. \cref{fig:waterfall} shows another comb graph with a designated root and many long teeth.\label{fig:combs}}
\end{figure}
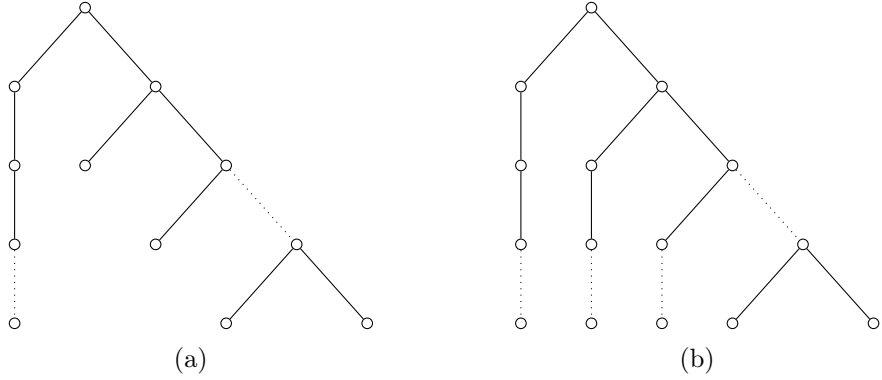

    The abusively labeled ``comb tree'' (see \cref{fig:comb}) helps us to demonstrate that SSMC can be exponentially faster than \cref{alg:GWW1*}, even for a relatively simple problem. Although the argument of this section can be generalized to other graphs like those in \cite{AV94}, we restrict our current attention to the relatively simple case depicted in \cref{fig:comb}. In particular, we assume that the depth of the tree is $D$, where the length of the ``tooth'' originating at depth $0$ is at most $D-1$ and every other tooth has length $1$. We further assume that each vertex on the spine has some tooth originating from it. (This can obviously generalize to cases that are not so restrictive.)
    
    \begin{theorem}\label{thm:comb}
        With $N$ walkers, \Cref{alg:GWW1*} has a probability of reaching the bottom right vertex of \cref{fig:comb} is at most $\bigO(N \exp(-D)) + \bigO(\exp(-N))$, where $D$ is the depth of the tree.
    \end{theorem}
    \begin{proof}
%        See \Cref{sec:comb}
        
    To compute the probability that \cref{alg:GWW1*} succeeds at reaching the deepest node of the comb tree (\cref{fig:comb}), we begin by supposing that we have a total population of $N$ walkers. Let $Y_j$ be the node on the left hand tooth at depth $j$. Because \cref{alg:GWW1*} always has $\eta_k(Y_j) = \eta_j(Y_j)\delta_{jk}$, we let $\eta(Y_j) = \eta_j(Y_j)$. We wish to derive a bound for \( P(\eta(Y_{d-1}) = N) \). 
    
    First, note that for $j < d$, $P(\eta(Y_{j+1}) < \eta(Y_j)) = 0$. That is, because \cref{alg:GWW1*} has all walkers on the left tooth continue to child nodes, there is no probability that any walker will jump to the spine unless $j=d$. 
    
    Where we do not not expect that it will cause confusion, we abusively write $\eta(v)$ for the set of walkers at vertex $v$, as well as the number of walkers at that vertex. Consider, now, the population $\eta(b_j)$ of walkers on the right hand branch $b_j$ at depth $j$. Note that each walker $w\in \eta(b_j)$ has a probability of $1/2$ of advancing to either of the right hand nodes at depth $j+1$. Denote by $l_j$ the leaf at depth $j$. After the walkers move, then probability that each walker at the leaf moves to the left hand branch is given by 
    \[
        \Pr{w \in \eta(b_{j+1}) | w \in \eta(b_j)} = \frac{\eta(Y_j)}{N-\eta(l_{j+1})} \geq \frac{\eta(Y_j)}{N}.
    \]
    Now, we find that
    \begin{align*}
        \Pr{w \in \eta(b_{j+1}) | w \in \eta(b_j)} &=1-\Pr{w \in \eta(b_{j+1}) | w \in \eta(l_{j+1})}\Pr{w \in \eta(l_{j+1}) \vert w \in \eta(b_j)} \\
        &\leq \left(1-\frac{\eta(Y_j)}{2 N}\right).
    \end{align*}
    Thus, the probability that a walker at depth $0$ is in $\eta(b_{j+1})$ is given by,
    \begin{align*}
        \Pr{w \in \eta(b_{j+1}) \vert w \in \eta(b_0)} &= \prod_{k=0}^{j} \Pr{w \in \eta(b_{k+1}) \vert w \in \eta(b_k)}  \\
        &\leq \prod_{k=1}^{j} \left(1-\frac{\eta(Y_j)}{2 N}\right) \\
        &\leq \left(1-\frac{\eta(Y_1))}{2 N}\right)^{j}\\
        &\leq \exp\left(-\frac{\eta(Y_1)}{2 N} j \right).
    \end{align*}
    Thus, 
    \begin{align*}
        \Pr{\eta(b_0) \subseteq \eta(Y_{j+1})} &\geq \left(1- \exp\left(-\frac{\eta(Y_1)}{2 N} j \right)\right)^{N-\eta(Y_1)} \\
        &\geq 1-\left(N-\eta(Y_1)\right)\exp\left(-\frac{\eta(Y_1)}{2 N} j \right).
    \end{align*}
    Hence, 
    \begin{equation}\label{eqn:store}
        \Pr{\eta(b_0) \subset \eta(Y_{j+1}) \vert \eta(Y_1) \geq y } \geq 1-(N-y)\exp\left(-\frac{y}{2 N} j \right).
    \end{equation}
    \begin{align*}
        \Pr{\eta(b_0) \subset \eta(Y_{j+1})} &\geq \sup_\delta \left(
            \Pr{\eta(b_0) \subset \eta(Y_{j+1}) \vert \eta(Y_1) \geq \frac{N}{2}-\delta}\Pr{\eta(Y_1) \geq \frac{N}{2} - \delta} \right)\\
            &\geq \sup_\delta \left(\Pr{\eta(b_0) \subset \eta(Y_{j+1}) \vert \eta(Y_1) \geq \frac{N}{2} - \delta} \left(1-e^{-2 \frac{\delta^2}{N}} \right)\right)\\
            &\geq \sup_\delta \left( 1-\left(\frac{N}{2}+\delta\right)\exp\left(-(1/4-\frac{\delta}{2N}) j \right)\right)\left(1-e^{-2 \frac{\delta^2}{N}} \right)
    \end{align*}
    where we have used Hoeffding's inequality and \cref{eqn:store}. For simplicity, take $\delta = c_1\frac{N}{2}$. Then, we have that
    \[
         \Pr{\eta(b_0) \subset \eta(Y_{j+1})} \geq (1-(1+c_1)\frac{N}{2}e^{-(1-c_1)\frac{j}{4}})(1-e^{-\frac{c_1^2 N}{4}}) = 1 -\bigO(Ne^{-j}) -\bigO(e^{-N})
    \]
    so that the probability that all walkers end on the left branch approaches 1 exponentially quickly in both depth and total number of walkers.
    
    Now, if the depth of the tree is $d+1$ and the depth of the left hand side is $j$, we have that 
    \begin{align*}
        \Pr{\eta(b_{d+1}) \neq \emptyset} &= \Pr{\eta(b_d) \neq \emptyset}\\
        &\leq 1 - \Pr{\eta(b_0) \subset \eta(Y_{j+1})}\\
        &\leq \mathcal{O}(N e^{-j}) + \mathcal{O}(e^{-N}).
    \end{align*}

    \end{proof}
     This example already points out a couple of problems with \cref{alg:GWW1*}. First, it is clear that the left hand side acts as a strong attractor for ``losing'' walkers. That is, in unbalanced search spaces, nonterminating paths will experience exponential growth. Second, one sees that \textit{additional walkers do not significantly increase the odds of arriving at the solution}. This is counter-intuitive, but because there is no mechanism to deplete the population on the left hand side, its population grows exponentially at the cost of the population on the right. Although the argument for \cref{thm:comb} does not directly address the issue, it is also clear that the probability of transitioning branches increases each time such a transition is made. 
     
\subsection{Paths to nowhere}
    In order to prove an exponential separation between \cref{alg:SSMC*} and \cref{alg:GWW1}, we consider a case where \cref{alg:GWW1} does no better than repeated random search. Consider the comb tree in \cref{fig:waterfall}. This is quite similar to the comb of the previous section, however each tooth of the comb now continues down to depth $D-1$. That is, we consider the comb with spine of length $D$ and where each tooth beginning at depth $j$ has length $D-1-j$.
    
    \begin{theorem}\label{thm:waterfall}
        With $N$ walkers, \Cref{alg:GWW1} has a probability of at most $\frac{N}{2^D}$ of reaching the bottom right vertex of \cref{fig:waterfall}, where $D$ is the depth of the tree. 
    \end{theorem}
    \begin{proof}
        Suppose we place $N$ walkers at the root. Then, because there are no leaves prior to depth $D$, walkers are never diverted from their current paths. The probability that any walker reaches the rightmost branch is simply $\frac{1}{2^D}$. Thus, for $N$ walkers, the probability that no walker reaches the rightmost branch is given by \((1-\frac{1}{2^D})^N \geq 1-\frac{N}{2^D} \).
    \end{proof}
     
     Hence, we know that the probability that \Cref{alg:GWW1} samples from the bottom right branch is exponentially small in the tree width.
     
     \begin{theorem}\label{thm:waterfall2}
         With the same conditions as in \Cref{thm:waterfall}, the probability that SSMC reaches the bottom right node of either tree in \cref{fig:combs} is at least $\frac{1}{D}-\frac{1}{N}$.
     \end{theorem}
     \begin{proof}
%         See \Cref{sec:waterfall2}

        Suppose that $b_i$ is the branch at level $i$. Then, let $B \subseteq V_{i+1}$ be the children of $b_i$. Note that no walkers at $b_i$ stay behind. Also, note that $m_{i+\half}(V_i) < 1$ by the construction of the algorithm. 
        \begin{align}
            \E{m_{i+1}(b_{i+1})\vert m_i(b_i)} &= \frac{1}{2}\E{m_{i+1}(B)\vert m_i(b_i)} \nonumber \\
            &= \frac{1}{2}\E{m_i(b_i) + m_i(b_i)\frac{m_{i+\half}(V_i)}{1-m_{i+\half}(V_i)} \vert m_i(b_i)} \nonumber\\
            &=\frac{m_i(b_i)}{2}\left(1+\E{\frac{m_{i+\half}(V_i)}{1-m_{i+\half}(V_i)} \vert m_i(b_i) }\right) \nonumber\\
            &\geq \frac{m_i(b_i)}{2}\left(1+\frac{\E{m_{i+\half}(V_i) \vert m_i(b_i)}}{1-\E{m_{i+\half}(V_i) \vert m_i(b_i)} }\right) \nonumber\\
            &= \frac{m_i(b_i)}{2}\left(1+\frac{\frac{1 - m_i(b_i)}{2}}{1-\frac{1 - m_i(b_i)}{2} }\right) \nonumber\\
            &=  \frac{m_i(b_i)}{1 + m_i(b_i)} \nonumber
        \end{align}
        where we have used the fact that the two vertices in $B$ have equal expected values and Jensen's inequality. To bound the decay rate, we define the new statistic, $\overline{m}_i = \min(m_i(b_i),\frac{1}{i+1})$. Note that $\overline{m}_i \leq m_i(b_i)$ and so conditioning on when this inequality is strict gives
        \begin{align*}
            \E{\overline{m}_{i+1}} &\geq \E{\frac{m_i(b_i)}{1 + m_i(b_i)} \vert m_i(b_i) 
            = \overline{m}_i}\Pr{m_i(b_i) 
            = \overline{m}_i} + \frac{1}{i+2}\Pr{m_i(b_i) 
            > \overline{m}_i} \\
            &= \E{\frac{\overline{m}_i}{1+\overline{m}_i}}\\
            &\geq \left(\frac{1+i}{2+i}\right)\E{\overline{m}_i}
        \end{align*}
        where the first inequality follows from the monotonicity of $\frac{u}{1+u}$ and the second inequality from increasing the denominator to its largest possible value.
        
        Solving the recurrence in $E(\overline{m}_i)$ yields
        \[
            \E{m_i(b_i)} \geq \E{\overline{m}_i} \geq \frac{1}{i+1}.
        \]        
        Now, we can apply Markov's inequality,
        \begin{align*}
            \Pr{m_{i}(b_i) \geq \frac{1}{N}} &\geq 1-\frac{1-\E{m_{i}(b_i)}}{1-\frac{1}{N}} \\
            &= 1-\frac{N-\frac{N}{i+1}}{N-1}
        \end{align*}
        and for $N \geq i$ we have the desired result,
        \[
            \Pr{m_{i}(b_i) \geq 1} \geq \frac{1}{i}-\frac{1}{N}.
        \]

     \end{proof}

\subsection{Quantum Annealing}
    Substochastic Monte Carlo was developed because of our desire for a better classical foil for quantum annealing (QA) than simulated annealing \cite{farhi2002quantum}. However, the results of this section show that, if viewed as a form of simulated QA, \textit{simulated quantum annealing can be exponentially faster than quantum annealing}. (We, of course, do not actually contend that this is simulated quantum annealing.) Such a strong claim requires a bit of a disclaimer, we do not claim that SSMC cannot be efficiently simulated by quantum annealing, but rather that when both processes are run using the same annealing schedule, SSMC can be exponentially faster. In other words, to properly achieve the results of SSMC with QA, one would need to implement a different annealing schedule. The results of \cite{Jarret2014a} suggest that the types of graphs studied here might be difficult, without simply implementing classical SSMC on quantum annealing hardware. Furthermore, our current results suggest that, even in more general search spaces, cases that are natively difficult for QA may still be solvable with SSMC. Our approach exploits the fact that natively incorporating structural information into a quantum annealing algorithm is difficult  \cite{Chancellor2016,Chancellor2017ModernizingSearches}. Because simulated quantum annealing typically seeks to simulate Schr\"{o}dinger evolution itself, this section also demonstrates a separation from the performance of simulated quantum annealing on these instances. 
    
    The quantum annealing algorithm evolves the initial distribution $\psi(0)$ by the Schr\"{o}dinger equation, instead of the heat equation, but otherwise our strategy remains the same as in the previous section. The quantum adiabatic theorem states that if we start in the ground state of $H(0)$, after some time $T$ we will end in a state $\psi(T)$ close to the ground state $H(T)$ as long as $T = \bigO(\gamma^{-2})$, where $\gamma$ is the minimal spectral gap of $H(t)$ with $t\in[0,T]$ \cite{elgart2012note,Jansen2006}. In our procedure, however, we evolve over a series of gapless Hamiltonians which encode something about the structure of the problem, and so the gap-based runtime estimate clearly fails. This does not guarantee that QA fails, rather just that our sufficient criterion for QA's success is not achieved. Under constraints outside the scope of this section, these algorithms can still potentially succeed \cite{avron1999adiabatic}. 
    
    To understand this in our context, we note that we evolve over a series of disconnected graphs. In solving the ``comb'' examples above, we always have the decomposition $H(t) = H_0(t) \oplus H_1(t)$. Then, if we let $\psi(t) = (\psi_0(t), \psi_1(t))$, we have that $H(t)\psi(t) = (H_0(t) \psi_0(t),H_1(t)\psi_1(t))$. That is, Schr\"{o}dinger evolution can be separated into the direct sum of two distinct Hilbert spaces, each evolved independently:
    \begin{equation*}
        i \frac{\partial \psi_{0,1}}{\partial t} = H_{0,1}(t)\psi_{0,1}.
    \end{equation*}
    The key point is that each space is independently norm-preserving. That is, $\norm{\psi_{0,1}(t)} = \norm{\psi_{0,1}(0)}$. Now, we consider the examples of the previous section. 
    
    The following propositions are all easily verified:
    \begin{proposition}\label{prop:aqc1}
        Suppose that $H(t) = H_0(t) \oplus H_1(t)$ for $t \in [j,j+1)$. Then, $\norm{\psi_{0,1}(j)}_2 = \norm{\psi_{0,1}(t)}_2$ for $t \in [j,j+1)$ where 
        \[
            i \frac{\partial \psi}{\partial t} = H(t)\psi
        \]
        and $\psi = (\psi_0,\psi_1)$.
    \end{proposition}
    That is, if $H$ separates into non-interacting subspaces, then the components of each subspace conserve norm independently. 
    
    \begin{proposition}\label{prop:aqc2}
        Suppose that $H(t) = a(t)L_j+b(t)W$ for $t\in[0,t_f]$, where $L_j$ is the graph Laplacian for a branch like in \cref{fig:graphs} with root $r$ and children $x,y$ and edge weights $w_{rx}=w_{ry}$. If
        \[
            i \frac{\partial \psi}{\partial t} = H(t)\psi
        \]
        and $\psi_x(0)=\psi_y(0)=0$, then
        \[
            \psi_x(t_f)= \psi_y(t_f) \leq \frac{1}{\sqrt{2}}\psi_r(0).
        \]
    \end{proposition}
    In words, because of symmetry, the maximum amount of square amplitude that can end up in either leaf is half of the total initial branch square amplitude. Thus, \Cref{prop:aqc1,prop:aqc2} combine to give the desired result.
    \begin{theorem}
        For the trees of \cref{fig:comb,fig:waterfall}, the quantum annealing algorithm run with the schedule of \cref{eqn:sched2} has amplitude at most 
        \[
            \norm{\psi(v)}_2 \leq \frac{1}{2^D}
        \]
        where $v$ is the bottom right vertex in either of \cref{fig:comb,fig:waterfall} and $D$ is depth of the tree.
    \end{theorem}
    In other words, even provided that the quantum annealing algorithm is capable of performing each local component in constant time, it still performs no better than repeated random search on these instances. Additionally, we have no guarantee that a different interpolation (that does not involve a classical process) will help with this particular search. Furthermore, \cite{Jarret2014a} suggests that creating such a strategy can sometimes be difficult. Thus, once structural information becomes relevant, we have that SSMC can be exponentially faster than quantum annealing \textit{with precisely the same annealing schedule}. 
    
    In situations such as this, \cite{Chancellor2016,coxson2014adiabatic} suggests the use of hybrid algorithms. For instance, one might use annealing in place of local search at each stage separately or for some other global procedure. This example, however, calls into question the utility of such techniques, since QA will not produce any more favorable statistics than local search. SSMC is indeed efficient locally and, because the particular, disconnected evolution disallows interference effects, annealing does not seem to add global utility. In fact, our speedup is precisely because SSMC utilizes ``infinite range'' interactions that quantum annealing cannot natively replicate. Nonetheless, one might simulate each walker locally with a quantum annealer in constant time and, thus, efficiently simulate SSMC on a quantum annealer. Whether or not this proposes a fundamental limit to QA on these search spaces is an interesting, open question.

\section{Beyond layered graphs}\label{sec:general}
As we stated previously, unlike GWW, SSMC is formulated for search on a general space. In this section, we demonstrate that SSMC is indeed capable of replicating the behavior of other search algorithms. In particular, we focus on gradient descent against biases. A biased walk is a simplified model of search in many spaces, such as flipping a particular bit in a long binary string. In this case, we have to turn away from the dynamics of \cref{eqn:sched2} and back to those of \cref{eqn:quantum}, the general SSMC setting.

    For an interpolation $H(t) = \frac{(1-t)}{d}L + t W$ with $L(t)$ with $t \in [0,1]$, $L$ the combinatorial Laplacian for an arbitrary graph of maximal degree $d$, and $0\leq W(v) \leq 1$ for all $v \in V$, we consider the small time estimate
    \[
        e^{-(\frac{1-t}{d}L+t W)\Delta t}\psi(t) \approx \left(I - \Delta t ( \frac{1-t}{d} L + t W) + \bigO(\Delta t^2)\right)\psi(t).
    \]
    and write 
    \begin{equation}\label{eqn:grad}
        T_{x\mapsto y} = \begin{cases}
            \left(I - \Delta t ( \frac{1-t}{d}L + t W\right)_{xy} & \text{$y \neq \infty$} \\
            \Delta t\; t W_{xx} & \text{$y = \infty$}            
        \end{cases}
    \end{equation}
    where we require that $\Delta t$ is taken sufficiently small such that $T$ restricted to the support of $\xi$ is a substochastic transition matrix. In other words, we let $\Delta t = \Delta t(t,\xi)$ be a nonlinear term. In particular, we consider the modified empirical process from \cref{sec:NumApprox}. That is, $W(v) \mapsto W(v) - \min_j W(\xi_j)$.  Let $\Delta E = \max_i W(\xi_i)-\min_i W(\xi_i)$, and, provided that $\Delta E \neq 0$, we have that
    \begin{equation}\label{eqn:deltat}
        \Delta t \leq \frac{1}{(1-t)+t\Delta E} \leq \frac{1}{\Delta E}.
    \end{equation}
    Now, under constraints that are always achieved over the linear interpolation, we can prove that SSMC will perform gradient descent against biases for sufficiently large $t,\Delta E$.
    \begin{proposition}\label{prop:descent}
        Consider a walk with $N>1$ walkers. Then, the probability that all walkers transition to states $\{v\}$ with $p= \min_i T_{\xi_i\mapsto v'\{v\}} < 1 - \max_i T_{\xi_i\mapsto \{v\}}$ and such that $W(u \in \{v\}) < W(\xi_i \sim u)$ for some $i$ satisfies
        \begin{align*}
            \Pr{m_{t+1}(v) = N \vert m_t} &\geq (1-e^{-N p})\left(1-N\frac{1-t}{t \Delta E}\right) \\&= 1 - \bigO\left(\frac{1-t}{t}\frac{N}{\Delta E}\right) - \bigO(e^{-N p}).
        \end{align*}
    \end{proposition}    
    \begin{proof}
%        See \cref{proof:propdescent}.

        We prove this in the case that all walkers are initialized to the same site $u$ and then sketch the proof for other distributions. Suppose that, initially, all walkers occupy vertex $u$. Then, $\Delta E = 0$. Thus, by \cref{eqn:grad}, no walkers transition to $\infty$. Now, let $T_{u\mapsto v} \leq 1 - T_{u\mapsto v}$ be the probability that walkers transition to site $v$, with $W(v) < W(u)$. Then, the probability that at least one walker jumped to state $v$ is given by
        \[
            \Pr{m_{t+1}(v) \neq 0} \geq 1 - (1 - T_{u\mapsto v})^N.
        \]
        Now, all walkers are distributed between two sites, where we have $T_{\xi_i \mapsto \infty} = \frac{\Delta E}{(1-t)+t \Delta E}$. Thus, 
        \[
            \Pr{m_{t+1+\half}(v) = 0 } \geq \left(\frac{\Delta E}{(1-t)+t \Delta E}\right)^N.
        \]
        Hence, 
        \begin{align*}
            \Pr{m_{t+2}(v) = N} &\geq \left(1 - (1 - T_{u\mapsto v})^N\right)\left(\frac{t \Delta E}{(1-t)+t \Delta E}\right)^N \\
            &\geq (1-e^{-N T_{u\mapsto v}})e^{-N\frac{1-t}{t \Delta E}}\\
            &\geq (1-e^{-N T_{u\mapsto v}})(1-N\frac{1-t}{t \Delta E})
        \end{align*}        
    
    To prove \cref{prop:descent} in full generality, one only needs to consider the behavior of the tail of the walkers occupying the highest site in the distribution. The proof simplifies, because the only behavior that needs to be considered up to $\bigO(\frac{N}{\Delta E})$ is the jump process. 

    \end{proof}
    In other words, despite the walk being biased against making the transition, the jump process still allows walkers to reliably transition to more preferable sites, provided such sites are locally available. This is a good model of, for instance, a walk on the hypercube with cost function equal to Hamming weight. Furthermore, the probability that walkers transition to lower states could be made less dependent on $N$ if we relax the number of walkers expected to make the transition.
    
    By taking $t$ to be sufficiently close to $1$, some stage of the algorithm always achieves gradient descent to whatever polynomial approximation we desire. That is, if $t\sim 1-\mathrm{poly}^{-1}(N)$, then we have that we perform gradient descent with probability $1 - \bigO(\frac{\mathrm{poly}^{-1}(N)}{\Delta E})$ for large enough $N$. Most of the time, however, we do not wish to simulate gradient descent anywhere but at the very end of an interpolation. That is, we wish to delay this behavior until after regions of local optima have been identified. 
    
    The theorem above is about an entirely general space and, thus, we do not require the restriction of SSMC to trees. In particular, we know that every interpolation of $H(t)$ will cross critical points such that behavior like gradient descent occurs. Furthermore, by continuity we know that there exist values of $t$ in the interpolation such that our procedure exhibits relative amounts of descent versus random search, including some value that should approximate an unbiased walk. That SSMC also crosses over regions with weak descent will allow the algorithm to cross barriers (such as those of \cite{Crosson2016SimulatedAnnealing}) with bounded probability, however we leave the analysis of these cases for future work. In particular, after some walkers clear a barrier, other remaining walkers will jump to positions across the barrier through the death process, precisely as occurs on combs. Hence we can see an exponential growth in walkers across barriers, allowing for an efficient continuation of search within difficult-to-locate regions.
   
%\nocite{LackeyGithub}

\section*{Acknowledgements}
We thank Stephen Jordan for useful discussions. \PIRA

\bibliography{main,mendeley}
 
\appendix

\section{Additional material}\label{sec:computations}

\subsection{Exact expression for transitions}

In the case of a node with one child, the dynamics of the continuous time process can be given in closed form. Here the transition rate matrix is given by
\begin{equation*}
    H = \begin{pmatrix} E + w & -w \\ -w & w \end{pmatrix}.
\end{equation*}

One can verify directly the two eigenvalues of this matrix are given by 
\begin{align*}
    \lambda_{\pm} = \lambda_2 + w \pm \Delta_2
\end{align*}    
where
\begin{align*}
    \lambda_2 &= \frac{E + w}{2} \\
    \Delta_2 &= \sqrt{\lambda_2^2 + 2 w^2} \\
    \Delta_2 &= \sqrt{\frac{E^2}{4} + \frac{wE}{2} + \frac{9w^2}{4}}.
\end{align*}
The associated eigenvectors, which we leave unnormalized, are 
\begin{equation*}
    \vec{v}_\pm = \sqrt{2}w\begin{pmatrix} \sqrt{2}w\\ \lambda_\mp - w \end{pmatrix} 
\end{equation*}
In particular,
$$\begin{pmatrix} 0\\1 \end{pmatrix} = \frac{1}{2\Delta_2}(\vec{v}_- - \vec{v}_+),$$
and 
$$\begin{pmatrix} 1\\0 \end{pmatrix} = \frac{1}{\sqrt{8}w}\left( (\vec{v}_+ + \vec{v}_-) - \frac{E+w}{2\Delta_2}(\vec{v}_- - \vec{v}_+)\right).$$

At a time $0 < t \leq 1$,
\begin{align*}
    e^{-Ht}\begin{pmatrix}1\\0\end{pmatrix} &= \frac{1}{\sqrt{8}w}\left[ (e^{-\lambda_+t}\vec{v}_+ + e^{-\lambda_-t}\vec{v}_-) - \frac{E+w}{2\Delta_2}(e^{-\lambda_-t}\vec{v}_- - e^{-\lambda_+t}\vec{v}_+)\right)\\
    &= e^{-(E + 3w)t/2}\begin{pmatrix} \cosh(\Delta_2t) - \frac{E+w}{2\Delta_2}\sinh(\Delta_2t) \\ \frac{w}{\sqrt{2}\Delta_2}\sinh(\Delta_2t)\end{pmatrix}.
\end{align*}
Writing this in coordinates gives
\begin{align*}
    \mathrm{Pr}\{X_{j+t} = x \:|\: X_j = x\} &= e^{-(E + 3w)t/2}\left(\cosh(\Delta_2t) - \frac{E+w}{2\Delta_2}\sinh(\Delta_2t)\right),\\
    \mathrm{Pr}\{X_{j+t} = y \:|\: X_j = x\} &= e^{-(E + 3w)t/2}\frac{w}{\Delta_2}\sinh(\Delta_2t).
\end{align*}
In particular in the case that $E \gg w$ we find
\begin{align*}
    \mathrm{Pr}\{X_{j+t} = x \:|\: X_j = x\} &= e^{-(w + \bigO(1/E))t}\left(\frac{2w^2}{E^2} + \bigO\left(\tfrac{1}{E^3}\right)\right),\\
    \mathrm{Pr}\{X_{j+t} = y_\mu \:|\: X_j = x\} &= e^{-(w + \bigO(1/E))t}\left(\frac{w}{E} - \frac{2w^2}{E^2} + \bigO\left(\tfrac{1}{E^3}\right)\right).
\end{align*}

\subsection{Derivation of the population dynamical equation}

First let express the dynamics of $\{X(t)\}_{t\geq 0}$ as a generator on functions $f:V(G)\cup\{\infty\} \to \mathbb{R}$. The key to linking the two is by taking $f = \mathds{1}_x$ for a given $x\in V(G)\cup\{\infty\}$ as then
$$E[f(X(t))] = \sum_y E[ f(X(t))\:|\:X(t) = y]\: \Pr{X(t) = y} = \sum_y f(y)\psi_y(t) = \psi_x(t)$$
when $f = \mathds{1}_x$. We are consider weighted graph Laplacians 
$$L_{y,x} = \left\{\begin{array}{cl} -w_{y,x} & \text{ when $y\not= x$,}\\ \sum_{y'} w_{y',x} & \text{ when $y=x$.}\end{array}\right.$$
The potential is given by $W_{y,x} = E_x\delta_{y,x}$. The ``death rate'' is then the column excess, $E_x$.

So,
\begin{align*}
    \frac{d}{dt} E[f(X(t))] &= \sum_y f(y)\frac{d\psi_y(t)}{dt} \\
    &= f(\infty)\frac{d\psi_\infty(t)}{dt} + \sum_{y\in V(G)} f(y)\frac{d\psi_y(t)}{dt}\\
    &= f(\infty)\left(\sum_x E_x \psi_x(t)\right) - \sum_{x,y\in V(G)} f(y) (L_{y,x} + W_{y,x}) \psi_x(t)\\
    &= f(\infty)\left(\sum_x E_x \psi_x(t)\right) - \sum_{x,y\in V(G)} f(y) (-w_{y,x} + \sum_{y'} w_{y',x}\delta_{yx} + E_x\delta_{yx}) \psi_x(t)\\
    &= \sum_{x\in V(G)} [(f(\infty) - f(x)) E_x + \sum_{y\in V(G)} w_{y,x}(f(y) - f(x))] \psi_x(t)\\
    &= -E[(\mathcal{H}f)(t)]
\end{align*}
where
$$(\mathcal{H}f)_x(t) = (f(x) - f(\infty)) E_x + \sum_{y\in V(G)} w_{y,x}(f(x) - f(y)).$$

And we can recover the dynamics of $X(t)$ by setting $f = \mathds{1}_y$ in the above:
\begin{align*}
    (\mathcal{H}\mathds{1}_y)_x(t) &= \delta_{xy}E_x + \sum_{z\in V(G)} w_{z,x}(\delta_{yx} - \delta_{xy})\\
    &= \delta_{yx}E_x - w_{y,x} + \delta_{yx}\sum_{x\in V(G)} w_{z,x}.
\end{align*}

Starting the analysis where all the walkers move independently on $V(G)\cup\{0\}$, we have for a function $f:(V(G)\cup\{\infty\})^N \to \mathbb{R}$ the transition rate generator acts as
\begin{align*}
    (\mathcal{H}_Nf)_{\vec{x}}(t) &= \sum_{i=1}^N (\mathcal{H}^{(i)}f)_{\vec{x}}(t)\\
    &= \sum_{i=1}^N [(f(\vec{x}) - (T_{x_i\mapsto\infty}f)(\vec{x}))E_{x_i} + \sum_{y\in V(G)} w_{y,x_i}(f(\vec{x}) - (T_{x_i\mapsto y}f)(\vec{x}))]
\end{align*}
where we write $T_{x_i\mapsto y}$ for the operator that substitutes $y$ for $x_i$ in position $i$. So to obtain the Fleming-Viot process we must replace $(T_{x_i\mapsto\infty}f)(\vec{x})$ the the corresponding term for a move to the site of another randomly selected walker. But this is easy: $\frac{1}{N-1}\sum_{j\not= i} (T_{x_i\mapsto x_j}f)(\vec{x})$. Therefore the generator for the Fleming-Viot process is
$$(\mathcal{H}_Nf)_{\vec{x}}(t) = \sum_{i=1}^N [(f(\vec{x}) - \tfrac{1}{N-1}\sum_{j\not= i} (T_{x_i\mapsto x_j}f)(\vec{x}))E_{x_i} + \sum_{y\in V(G)} w_{y,x_i}(f(\vec{x}) - (T_{x_i\mapsto y}f)(\vec{x}))].$$

To clarify this exposition, we will view our statistic $\eta$ is a different light. Fix an $x\in V(G)$ and take $\eta_x(\vec{x}) = \sum_{i=1}^N \mathds{1}_{x_i = x}$. As $\eta_x$ counts the number of $x_i$ that equal $x$, the function in the main text satisfies $\eta(t;x) = \eta_x(\xi(t))$. Now we compute
\begin{equation*}
    (T_{x_i\mapsto x_j}\eta_x)(\vec{x}) = \left\{\begin{array}{cl} \eta_x(\vec{x}) + 1 & x_i\not= x,\ x_j = x\\ \eta_x(\vec{x}) - 1 & x_i = x,\ x_j\not= x\\ \eta_x(\vec{x}) & \text{otherwise.}\end{array}\right.
\end{equation*}
And so,
\begin{align*}
    \tfrac{1}{N-1} \sum_{j \not= i} (T_{x_i\mapsto x_j}\eta_x)(\vec{x}) 
    &= \frac{1}{N-1}\left\{\begin{array}{cl} (\eta_x(\vec{x}) + 1)\eta_x(\vec{x}) + \eta_x(\vec{x})(N-1-\eta_x(\vec{x})) & x_i\not= x\\ \eta_x(\vec{x})(\eta_x(\vec{x}) - 1) + (\eta_x(\vec{x})-1)(N-\eta_x(\vec{x})) & x_i= x.\end{array}\right.\\
    &= \frac{1}{N-1}\left\{\begin{array}{cl} N\eta_x(\vec{x}) & x_i\not= x\\ (N-1)\eta_x(\vec{x}) - N & x_i= x.\end{array}\right.
\end{align*}
This gives
\begin{align*}
    &\sum_{i=1}^N (f(\vec{x}) - \tfrac{1}{N-1}\sum_{j\not= i} (T_{x_i\mapsto x_j}f)(\vec{x}))E_{x_i}\\
    &\quad = \eta_x(\vec{x})E_x - \tfrac{1}{N-1}\sum_{y\not= x} \eta_x(\vec{x})\eta_y(\vec{x})E_y.
\end{align*}
A similar computation shows
\begin{align*}
    &\sum_{i=1}^N \sum_{y\in V(G)} w_{y,x_i}(f(\vec{x}) - (T_{x_i\mapsto y}f)(\vec{x}))\\
    &\quad = \sum_{y\not= x} w_{y,x}\eta_x(\vec{x}) - w_{x,y}\eta_y(\vec{x}).
\end{align*}
Therefore
$$(\mathcal{H}_N\eta_x)(t) = \sum_{y\not= x} (w_{y,x}\eta_x - w_{x,y}\eta_y - \tfrac{1}{N-1}\eta_x\eta_y E_y) + \eta_xE_x.$$

\end{document}